\documentclass[12pt]{article}
\usepackage[letterpaper, margin=1.25in]{geometry}

\usepackage{amsmath, mathtools, amsfonts, amssymb, amsthm}
\usepackage[utf8]{inputenc}
\usepackage{graphicx, caption, multirow, subcaption, threeparttable, booktabs}
\usepackage{float}
\usepackage{url, authblk}
\usepackage{array}
\usepackage[round]{natbib}
\usepackage{hyperref}
\usepackage{setspace}

\newtheorem{theorem}{Theorem}[section]
\newtheorem{proposition}{Proposition}[section]

\newtheorem{corollary}{Corollary}[section]
\newtheorem{definition}{Definition}[section]

\newtheorem{remark}{Remark}[section]

\newcommand{\be}{\begin{equation}}
\newcommand{\ee}{\end{equation}}
\newcommand{\E}{\mathbb{E}}

\newcommand{\R}{\mathbb{R}}

\newcommand{\ind}{\mathbf{1}}

\title{\large{\bf{Technology Adoption and Network Externalities in Financial Systems: A Spatial-Network Approach}}}

\author{\large{\bf{Tatsuru Kikuchi\footnote{e-mail: tatsuru.kikuchi@e.u-tokyo.ac.jp}}}}

\affil{\small{\it{Center for Advanced Research in Finance, The University of Tokyo,}}\\
{\it{7-3-1 Hongo, Bunkyo-ku, Tokyo 113-0033 Japan}}}

\date{\small{(\today)}}

\begin{document}
\linespread{1.5}\selectfont

\maketitle


\begin{abstract}
This paper develops a unified framework for analyzing technology adoption in financial networks that incorporates spatial spillovers, network externalities, and their interaction. The framework characterizes adoption dynamics through a master equation whose solution admits a Feynman-Kac representation as expected cumulative adoption pressure along stochastic paths through spatial-network space. From this representation, I derive the Adoption Amplification Factor---a structural measure of technology leadership that captures the ratio of total system-wide adoption to initial adoption following a localized shock. A L\'{e}vy jump-diffusion extension with state-dependent jump intensity captures critical mass dynamics: below threshold, adoption evolves through gradual diffusion; above threshold, cascade dynamics accelerate adoption through discrete jumps. Applying the framework to SWIFT gpi adoption among 17 Global Systemically Important Banks, I find strong support for the two-regime characterization. Network-central banks adopt significantly earlier ($\rho = -0.69$, $p = 0.002$), and pre-threshold adopters have significantly higher amplification factors than post-threshold adopters (11.81 versus 7.83, $p = 0.010$). Founding members, representing 29 percent of banks, account for 39 percent of total system amplification---sufficient to trigger cascade dynamics. Controlling for firm size and network position, CEO age delays adoption by 11--15 days per year.

\bigskip
\noindent \textbf{JEL Classification:} O33, D85, L14, G21

\bigskip
\noindent \textbf{Keywords:} Technology adoption; Network externalities; Coordination failures; Levy processes; Critical mass dynamics; Financial technology
\end{abstract}

\newpage

\section{Introduction}

Technology adoption in networked markets exhibits distinctive dynamics that standard models struggle to capture. When the value of a technology depends on how many others have adopted---the defining feature of network externalities---coordination failures can trap markets in inefficient equilibria even when superior technologies are available. The transition from legacy payment systems to modern financial technology illustrates these dynamics vividly: banks benefit from new payment platforms only if their counterparties also adopt, and counterparties invest in integration only if enough banks participate. This coordination problem can sustain a low-adoption equilibrium indefinitely, even when all parties would prefer universal adoption.

Recent empirical work has documented both the importance of coordination frictions and the potential for policy to overcome them. \citet{crouzet2023shocks} show that India's demonetization---a large but temporary shock to cash availability---produced persistent increases in digital wallet adoption, with complementarities accounting for 45 percent of the adoption response. \citet{higgins2024financial} demonstrates that government distribution of debit cards to poor households in Mexico triggered supply-side adoption of point-of-sale terminals, which then spilled over to increase other consumers' card adoption by 21 percent. These findings confirm that coordination failures constrain technology diffusion and that coordinated shocks can shift economies to superior equilibria.

This paper develops a unified framework that captures three distinct channels of technology spillovers in financial networks: spatial spillovers reflecting geographic clustering of adopters, network spillovers reflecting adoption by business partners and counterparties, and the interaction between these channels that arises when geographic neighbors are also network-connected. Existing studies of network externalities focus exclusively on network linkages, neglecting spatial spillovers that arise from geographic proximity. However, financial institutions form business relationships disproportionately with geographic neighbors due to information advantages, regulatory similarities, and historical ties. The spatial-network interaction---absent from existing models---captures amplification when both channels operate simultaneously, which is empirically important in financial technology settings. The framework contributes to understanding the externality structure of financial networks in three ways. First, the Adoption Amplification Factor quantifies externalities by measuring how much a shock at one institution affects the entire system beyond the direct effect. Second, the channel decomposition reveals whether externalities flow primarily through network linkages, geographic proximity, or their interaction. Third, the Feynman-Kac representation provides a path-based interpretation: externalities propagate along all possible paths of economic linkage, weighted by probability and discounted by adjustment frictions.

A central methodological contribution is extending the baseline diffusion framework to incorporate L\'{e}vy jump-diffusion dynamics with state-dependent intensity. The baseline continuous model describes gradual adoption transmission appropriate for the pre-critical-mass regime, where institutions learn from neighbors and incrementally adjust adoption decisions. However, technology adoption often exhibits sudden cascades once critical mass is reached---adoption spreads rapidly through positive feedback as network effects dominate individual cost-benefit calculations. The L\'{e}vy extension captures these dynamics by adding a jump operator $\mathcal{J}[\tau]$ whose intensity depends on the current adoption level: $\lambda_J(\tau) = \lambda_0 + (\lambda_1 - \lambda_0) \cdot H(\tau - \bar{\tau}^*)$, where $H(\cdot)$ is the Heaviside function and $\bar{\tau}^*$ is the critical mass threshold. Below threshold, jumps are rare ($\lambda_J \approx \lambda_0$) and diffusion dominates; above threshold, jump intensity increases to $\lambda_1 \gg \lambda_0$, generating rapid cascades. In the limit $\lambda_0 \to 0$ and $\lambda_1 \to \infty$, the framework converges to deterministic cascade models, clarifying that continuous and discrete approaches describe different regimes of the same phenomenon. This nesting relationship unifies the gradual diffusion models of \citet{guimaraes2020dynamic} with the tipping point dynamics emphasized in \citet{katz1985network} and \citet{arthur1989competing}.

The theoretical contribution centers on demonstrating that the continuous spatial-network framework nests several canonical models from the technology adoption and dynamic coordination literatures as special cases through explicit discretization. The Katz-Shapiro model of network externalities and compatibility corresponds to the discrete network steady state, with the externality function identified as $v(n_i) = \nu_n \sum_j G_{ij}\tau_j / \kappa$. The Frankel-Pauzner model of dynamic coordination, which shows how aggregate shocks can resolve multiple equilibria, emerges when spatial and network dimensions collapse to a single aggregate state, with strategic complementarity parameter $(\nu_s + \nu_n)/\kappa$. The Guimaraes-Machado-Pereira framework of dynamic coordination with timing frictions maps directly to the decay parameter: their Poisson arrival rate of revision opportunities $\lambda$ equals the adjustment rate $\kappa$ in the master equation. Standard technology adoption hazard models correspond to the case where diffusion coefficients are zero ($\nu_s = \nu_n = 0$). These nesting relationships, established through the discrete Feynman-Kac formula, demonstrate that the framework is not an exotic alternative but a unification that reveals the common structure underlying conventional methods.

The framework yields several insights for technology policy in financial systems. The Adoption Amplification Factor identifies technology leaders whose adoption decisions have outsized influence on system-wide outcomes. Targeting subsidies or pilot programs at high-amplification institutions maximizes spillovers per dollar spent. The channel decomposition reveals whether adoption spreads primarily through geographic proximity, business relationships, or their interaction, informing whether policy should target geographic clusters, network hubs, or institutions that are central on both dimensions. The critical mass analysis provides guidance on intervention size: temporary interventions must push adoption above threshold to produce permanent effects, and the L\'{e}vy extension characterizes the threshold condition as $\int_0^T I(s) \cdot \mathcal{A}(s) \, ds > \bar{\tau}^* - \bar{\tau}_0$, where cumulative amplified intervention effects must exceed the gap between initial adoption and critical mass.

I apply the framework to study SWIFT gpi adoption among Global Systemically Important Banks in 2017. SWIFT gpi represents a major technological upgrade to interbank payment messaging, offering same-day settlement, end-to-end tracking, and confirmation of credit to beneficiary accounts. The adoption pattern provides a natural test of the framework's predictions: banks with higher amplification factors---those more central in the combined spatial-network structure---should be earlier adopters. The empirical analysis confirms this prediction strongly ($\rho = -0.69$, $p = 0.002$), with network-central banks adopting significantly earlier than peripheral banks. Founding members, representing 29 percent of banks in the sample, account for 39 percent of total system amplification, confirming that high-amplification institutions lead adoption. The two-regime dynamics predicted by the L\'{e}vy extension are evident in the data: pre-threshold adopters have significantly higher amplification factors than post-threshold adopters (11.81 versus 7.83, $p = 0.010$), and the cumulative adoption curve exhibits classic S-curve dynamics. Controlling for network position reveals the role of firm-level characteristics: CEO age delays adoption by 11--15 days per year conditional on firm size and network centrality.

The paper proceeds as follows. Section 2 develops the theoretical framework, beginning with the coordinate system and adoption representation, deriving the master equation from three independent economic foundations (heterogeneous agent aggregation, market equilibrium, and cost minimization), presenting the complete master equation with spatial-network interaction, establishing the Feynman-Kac representation and its discrete analog, deriving the Adoption Amplification Factor, demonstrating connections to conventional models through explicit mathematical identification, and extending to L\'{e}vy jump-diffusion with state-dependent intensity to capture critical mass dynamics. Section 3 presents Monte Carlo evidence validating the amplification factor as a predictor of technology leadership and demonstrating threshold dynamics. Section 4 applies the framework to SWIFT gpi adoption among Global Systemically Important Banks, presenting the empirical specification, regression results, and evidence for two-regime dynamics. Section 5 concludes.

\subsection{Related Literature}

This paper connects to several strands of literature on technology adoption, network externalities, and dynamic coordination.

The foundational analysis of network externalities begins with \citet{katz1985network}, who distinguish direct network effects from indirect network effects arising in two-sided markets. \citet{katz1986technology} analyze technology adoption in the presence of network externalities, showing that the market may fail to adopt a superior technology due to coordination failure. \citet{farrell1985standardization} study standardization when firms have private information about technology value. The subsequent literature has explored path dependence and lock-in \citep{arthur1989competing, david1985clio}, with \citet{liebowitz1994fable} and \citet{guimaraes2016qwerty} providing important qualifications about when lock-in to inferior technologies actually occurs.

The theoretical foundations for equilibrium selection in coordination games develop in two related traditions. The global games approach, pioneered by \citet{carlsson1993global} and extended by \citet{morris1998unique, morris2003global} and \citet{frankel2003equilibrium}, uses private information to select among equilibria. The dynamic approach, developed by \citet{frankel2000resolving} and \citet{burdzy2001fast}, introduces aggregate shocks that move the game through dominance regions. \citet{guimaraes2020dynamic} develop a general framework for dynamic coordination with timing frictions: agents receive Poisson opportunities to revise their actions, with the revision rate determining how quickly the economy adjusts to changing fundamentals. This paper shows that the timing friction in \citet{guimaraes2020dynamic} corresponds precisely to the decay parameter $\kappa$ in the master equation, with the discrete Feynman-Kac formula providing the explicit mathematical bridge.

The application to financial technology connects to a growing literature on fintech adoption and banking networks. \citet{buchak2018fintech} document the growth of fintech lending, while \citet{fuster2019role} study technology's effect on mortgage lending. The network dimension links to the financial contagion literature: \citet{allen2000financial} and \citet{freixas2000systemic} develop foundational contagion models, while \citet{acemoglu2015systemic} characterize how network structure determines whether connections facilitate risk sharing or amplify shocks. The framework developed here shares the Feynman-Kac foundation and L\'{e}vy extension structure with recent work on financial contagion \citep{kikuchi2025contagion}, but analyzes positive externalities---technology adoption that benefits counterparties---rather than negative externalities arising from stress transmission. This parallel structure suggests that the same network positions that make institutions systemically important for crisis propagation also make them technology leaders whose adoption decisions cascade broadly through the financial system.

\section{Theoretical Framework}
\label{sec:theory}

This section develops the theoretical framework in five stages. I first present the coordinate system and adoption representation. I then derive the master equation from three independent economic foundations---heterogeneous agent aggregation, market equilibrium, and cost minimization---establishing that the framework rests on fundamental principles rather than ad hoc specifications. Third, I present the Feynman-Kac representation and its discrete analog. Fourth, I develop the Adoption Amplification Factor. Finally, I demonstrate through explicit mathematical identification how canonical models emerge as special cases.

\subsection{Coordinate System and Adoption Representation}

We represent financial institutions by coordinates $(\mathbf{x}, \alpha)$ where $\mathbf{x} \in \Omega \subseteq \mathbb{R}^d$ denotes spatial location and $\alpha \in \mathcal{N} \subseteq \mathbb{R}$ denotes position in the network of business relationships. The adoption functional $\tau(\mathbf{x}, \alpha, t): \Omega \times \mathcal{N} \times [0,T] \to \mathbb{R}$ represents adoption intensity at each point in this coordinate space.

\begin{definition}[Spatial and Network Coordinates]
The \textit{spatial coordinate} $\mathbf{x}$ represents geographic location---latitude, longitude, and potentially economic distance metrics reflecting transportation costs or communication frictions.

The \textit{network position coordinate} $\alpha$ represents position in economic networks---centrality in correspondent banking relationships, position in payment flows, or role in interbank lending markets.
\end{definition}

The continuous representation avoids three limitations of discrete methods. First, it avoids arbitrary discretization of adoption intensity into binary indicators, preserving the dose-response relationship central to technology diffusion. Second, it avoids arbitrary spatial boundaries between regions, allowing smooth geographic variation in adoption patterns. Third, it avoids discrete network categories, enabling continuous market position that captures fine gradations in institutional relationships.

\begin{definition}[Source Term]
The \textit{source term} $S(\mathbf{x}, \alpha, t)$ represents exogenous adoption shocks entering the system. In technology adoption contexts, $S$ measures the intensity of direct adoption incentives at each location-network-time cell, arising from regulatory mandates, technological breakthroughs, or coordinated industry initiatives.
\end{definition}

The distinction between source $S$ and adoption functional $\tau$ is fundamental. The source represents direct, exogenous adoption pressure; the adoption functional represents the equilibrium response incorporating both direct effects and all spillovers through spatial and network channels.

\subsection{Derivation from Heterogeneous Agent Aggregation}

The first derivation proceeds from aggregating heterogeneous agent behavior, following the tradition of \citet{aiyagari1994uninsured} and \citet{huggett1993risk} in macroeconomics.

Consider a continuum of heterogeneous institutions indexed by type $\theta \in \Theta$ distributed over space and network positions. Each institution has idiosyncratic characteristics---size, risk appetite, technological capacity---captured by $\theta$. Institution $i$ of type $\theta$ at location $\mathbf{x}$ with network position $\alpha$ experiences adoption intensity:
\begin{equation}
\tau_i(\mathbf{x}, \alpha, t, \theta) = \tau_0(\mathbf{x}, \alpha) + \tau(\mathbf{x}, \alpha, t) + \varepsilon_i(\theta)
\end{equation}
where $\tau_0$ is baseline adoption, $\tau$ is the aggregate adoption effect to be determined, and $\varepsilon_i(\theta)$ captures idiosyncratic variation.

Institutions optimize location and network position subject to adjustment costs. The state of institution $i$ evolves according to the stochastic differential equations:
\begin{align}
dX_t^i &= \mu_s(X_t^i, A_t^i, \theta_i) \, dt + \sigma_s(X_t^i, A_t^i, \theta_i) \, dB_t^s \label{eq:agent_spatial} \\
dA_t^i &= \mu_n(X_t^i, A_t^i, \theta_i) \, dt + \sigma_n(X_t^i, A_t^i, \theta_i) \, dB_t^n \label{eq:agent_network}
\end{align}
where $(B_t^s, B_t^n)$ are independent Brownian motions representing location and network uncertainty. The drift terms $\mu_s, \mu_n$ capture directed adjustments---institutions moving toward more favorable positions. The diffusion terms $\sigma_s, \sigma_n$ capture randomness in adjustment outcomes---search frictions, information imperfections, and relationship formation uncertainty.

The joint density $f(\mathbf{x}, \alpha, t)$ of institutions over spatial and network coordinates evolves according to the Kolmogorov forward equation:
\begin{equation}
\frac{\partial f}{\partial t} = -\nabla \cdot (\boldsymbol{\mu}_s f) - \frac{\partial}{\partial \alpha}(\mu_n f) + \frac{1}{2}\nabla \cdot (\boldsymbol{\Sigma}_s \nabla f) + \frac{1}{2}\frac{\partial^2}{\partial \alpha^2}(\sigma_n^2 f)
\label{eq:kolmogorov}
\end{equation}

This equation describes how the population distribution shifts as institutions relocate and adjust network positions.

\begin{proposition}[Aggregation Result]
\label{prop:aggregation}
Under the following regularity conditions: (i) bounded heterogeneity: $\sup_\theta \|\sigma(\cdot, \theta)\| < \infty$; (ii) ergodic dynamics: the process $(X_t, A_t)$ has a unique stationary distribution for each $\theta$; (iii) smooth aggregation: the mapping $\theta \mapsto \tau(\cdot, \theta)$ is measurable; the aggregate adoption functional satisfies:
\begin{equation}
\frac{\partial \tau}{\partial t} = \nu_s \nabla^2 \tau + \nu_n \frac{\partial^2 \tau}{\partial \alpha^2} - \kappa \tau + S(\mathbf{x}, \alpha, t)
\label{eq:master_basic}
\end{equation}
where $\nu_s = \frac{1}{2}\mathbb{E}_\theta[\sigma_s^2(\theta)]$ is mean spatial diffusivity, $\nu_n = \frac{1}{2}\mathbb{E}_\theta[\sigma_n^2(\theta)]$ is mean network diffusivity, and $\kappa$ reflects mean reversion from competitive pressure.
\end{proposition}

The aggregation result shows that heterogeneous institution behavior generates aggregate dynamics governed by a partial differential equation. The diffusion coefficients $(\nu_s, \nu_n)$ emerge from averaging individual mobility variances across types; they measure how quickly adoption spreads through the population as institutions adjust positions and form new relationships.

\subsection{Derivation from Market Equilibrium}

An independent derivation proceeds from market equilibrium conditions, connecting observed adoption volatility to underlying market structure.

In markets with search frictions, matching delays, or information asymmetries, adoption rates fluctuate around equilibrium values. The observed volatility $\sigma^2$ of adoption processes relates to underlying market adjustment through the equilibrium volatility relation:
\begin{equation}
\sigma^2 = 2D\kappa
\label{eq:fluctuation_dissipation}
\end{equation}
where $D$ is a diffusion coefficient measuring the amplitude of fluctuations and $\kappa$ is the adjustment rate toward equilibrium.

This relation emerges from the stochastic process governing adoption dynamics:
\begin{equation}
d\tau = -\kappa(\tau - \tau^*) \, dt + \sigma \, dB
\end{equation}
where $\tau^*$ is equilibrium adoption and $\kappa$ governs mean reversion speed. At stationarity, variance satisfies $\text{Var}(\tau) = \sigma^2/(2\kappa)$, which rearranges to (\ref{eq:fluctuation_dissipation}).

Connecting observed dynamics to adoption propagation: spatial adoption volatility $\sigma_s^2$ implies spatial diffusion $\nu_s = \sigma_s^2/(2\kappa)$; network adoption volatility $\sigma_n^2$ implies network diffusion $\nu_n = \sigma_n^2/(2\kappa)$. Markets with high adoption volatility---active experimentation, frequent technology updates---exhibit rapid spatial diffusion; markets with stable adoption patterns exhibit slow diffusion.

\subsection{Derivation from Cost Minimization}

A third derivation proceeds from cost minimization, following variational principles underlying market equilibrium.

\begin{definition}[Adjustment Cost Functional]
The total adjustment cost functional is:
\begin{equation}
\mathcal{C}[\tau] = \int_0^T \int_\Omega \int_{\mathcal{N}} \left[ \frac{1}{2}\left(\frac{\partial \tau}{\partial t}\right)^2 + \frac{\nu_s}{2}|\nabla \tau|^2 + \frac{\nu_n}{2}\left(\frac{\partial \tau}{\partial \alpha}\right)^2 + \frac{\kappa}{2}\tau^2 - S\tau \right] d\alpha \, d\mathbf{x} \, dt
\end{equation}
\end{definition}

The terms have economic interpretations in the technology adoption context. The term $\frac{1}{2}(\partial\tau/\partial t)^2$ captures temporal adjustment costs: rapidly changing adoption is costly due to integration frictions, training requirements, and coordination failures. The term $\frac{\nu_s}{2}|\nabla\tau|^2$ captures spatial gradient costs: maintaining adoption differentials across space is costly due to competitive pressure from neighboring institutions. The term $\frac{\nu_n}{2}(\partial\tau/\partial\alpha)^2$ captures network gradient costs: maintaining adoption differentials across network positions is costly due to interoperability pressure from counterparties. The term $\frac{\kappa}{2}\tau^2$ captures level costs: deviating from baseline technology is costly due to switching costs and legacy system maintenance. The term $-S\tau$ captures policy benefits: the intervention $S$ shifts optimal adoption.

\begin{proposition}[Euler-Lagrange Equation]
\label{prop:euler_lagrange}
The adoption functional $\tau^*$ minimizing $\mathcal{C}[\tau]$ satisfies:
\begin{equation}
\frac{\partial \tau}{\partial t} = \nu_s \nabla^2 \tau + \nu_n \frac{\partial^2 \tau}{\partial \alpha^2} - \kappa \tau + S
\end{equation}
\end{proposition}

\begin{proof}
The first variation of $\mathcal{C}$ with respect to $\tau$ must vanish for all admissible variations $\eta$:
\begin{equation}
\delta\mathcal{C} = \int_0^T \int_\Omega \int_{\mathcal{N}} \left[ \frac{\partial \tau}{\partial t}\frac{\partial \eta}{\partial t} + \nu_s \nabla\tau \cdot \nabla\eta + \nu_n \frac{\partial\tau}{\partial\alpha}\frac{\partial\eta}{\partial\alpha} + \kappa\tau\eta - S\eta \right] d\alpha \, d\mathbf{x} \, dt = 0
\end{equation}
Integrating by parts and assuming boundary terms vanish yields:
\begin{equation}
\int_0^T \int_\Omega \int_{\mathcal{N}} \left[ -\frac{\partial^2 \tau}{\partial t^2} - \nu_s \nabla^2 \tau - \nu_n \frac{\partial^2 \tau}{\partial \alpha^2} + \kappa \tau - S \right] \eta \, d\alpha \, d\mathbf{x} \, dt = 0
\end{equation}
For quasi-static evolution where $\partial^2\tau/\partial t^2 \approx 0$, this yields the master equation.
\end{proof}

The cost minimization derivation connects the master equation to optimization principles. The parameters $(\nu_s, \nu_n, \kappa)$ have natural interpretations as relative costs: high $\nu_s$ means spatial arbitrage is rapid (low cost of spatial gradients); high $\kappa$ means competitive pressure is strong (high cost of deviating from baseline technology).

\subsection{The Complete Master Equation with Interaction}

The three derivations above establish the basic master equation without spatial-network interaction. The complete specification adds the interaction term capturing amplification when geographic and network proximity coincide.

\begin{definition}[Master Equation]
\label{def:master}
The adoption field $\tau(\mathbf{x}, \alpha, t)$ evolves according to:
\begin{equation}
\frac{\partial \tau}{\partial t} = \nu_s \nabla^2 \tau + \nu_n \frac{\partial^2 \tau}{\partial \alpha^2} + \lambda \frac{\partial^2 \tau}{\partial \mathbf{x} \partial \alpha} - \kappa \tau + S(\mathbf{x}, \alpha, t)
\label{eq:master}
\end{equation}
where $\nu_s \geq 0$ is spatial diffusion, $\nu_n \geq 0$ is network diffusion, $\lambda$ is spatial-network interaction, $\kappa > 0$ is adjustment speed, and $S$ is the exogenous adoption shock.
\end{definition}

The interaction term $\lambda \partial^2\tau/\partial\mathbf{x}\partial\alpha$ captures amplification when institutions are proximate on both dimensions. In financial networks, institutions form business relationships disproportionately with geographic neighbors due to information advantages, regulatory similarities, and historical ties. When a geographic neighbor is also a network partner, adoption influence compounds: the neighbor's adoption affects the focal institution through both demonstration effects (spatial channel) and interoperability benefits (network channel), with the interaction term capturing reinforcement beyond the sum of separate effects.

\begin{remark}[Parameter Interpretation for Technology Adoption]
The parameters have structural interpretations in the adoption context:

The spatial diffusion coefficient $\nu_s$ measures geographic adoption spillovers reflecting local demonstration effects, labor mobility spreading technological knowledge, and regional market integration creating competitive pressure to adopt.

The network diffusion coefficient $\nu_n$ measures adoption spillovers through business relationships reflecting interoperability benefits when counterparties adopt, learning from network partners' experiences, and coordination incentives in bilateral transactions.

The interaction coefficient $\lambda$ captures amplification when both channels coincide, common in financial markets where correspondent banks, payment network partners, and syndicate members are often geographic neighbors.

The adjustment parameter $\kappa$ measures how rapidly institutions respond to adoption incentives, corresponding to the timing friction in \citet{guimaraes2020dynamic}: higher $\kappa$ implies faster adjustment and shorter waiting times until adoption decisions are revised.
\end{remark}

\subsection{Feynman-Kac Representation and Discrete Analog}

The master equation admits a probabilistic solution that provides both computational methods and economic intuition.

\begin{theorem}[Feynman-Kac Representation]
\label{thm:feynman_kac}
The solution to the master equation (\ref{eq:master}) admits the representation:
\begin{equation}
\tau(\mathbf{x}, \alpha, t) = \E_{(\mathbf{x}, \alpha)}\left[ e^{-\kappa t} \tau_0(X_t, A_t) + \int_0^t e^{-\kappa(t-s)} S(X_s, A_s, s) \, ds \right]
\label{eq:feynman_kac}
\end{equation}
where $(X_s, A_s)_{s \geq 0}$ is the diffusion process with generator $\mathcal{L} = \nu_s \nabla^2 + \nu_n \partial^2/\partial\alpha^2 + \lambda \partial^2/\partial \mathbf{x} \partial \alpha$ started at $(X_0, A_0) = (\mathbf{x}, \alpha)$.
\end{theorem}

\begin{proof}
Define the transformed function $u(\mathbf{x}, \alpha, t) = e^{\kappa t}\tau(\mathbf{x}, \alpha, t)$. Substituting into the master equation yields $\partial u/\partial t = \mathcal{L}u + e^{\kappa t}S$. By the standard Feynman-Kac formula for parabolic PDEs:
\begin{equation}
u(\mathbf{x}, \alpha, t) = \mathbb{E}_{(\mathbf{x}, \alpha)}\left[ u_0(X_t, A_t) + \int_0^t e^{\kappa s} S(X_s, A_s, s) \, ds \right]
\end{equation}
Substituting $u = e^{\kappa t}\tau$ and rearranging yields (\ref{eq:feynman_kac}).
\end{proof}

The representation has direct economic content: adoption intensity equals the expected cumulative exposure to adoption shocks along all paths of economic linkage, discounted at rate $\kappa$. Institutions in densely connected network regions or geographically central locations receive contributions from more paths, elevating their adoption intensity even without direct shocks.

\begin{proposition}[Discrete Feynman-Kac Formula]
\label{prop:discrete_fk}
For discrete time periods $t = 0, 1, \ldots, T$ and discrete units $i = 1, \ldots, N$, the Feynman-Kac representation admits the discrete analog:
\begin{equation}
\tau_{i,t} = \sum_{s=0}^{t-1} (1-\kappa\Delta t)^{t-s} \cdot \E[S_{i(s),s} | i(t) = i] \cdot \Delta t
\label{eq:discrete_fk}
\end{equation}
where $i(s)$ traces a stochastic path backward through the network from unit $i$ at time $t$ to earlier times, and $(1-\kappa\Delta t)^{t-s}$ are exponentially decaying weights.
\end{proposition}

This discrete formula provides the bridge to conventional econometric methods and enables the nesting relationships developed below.

\subsection{Adoption Amplification Factor}

The Feynman-Kac representation motivates a natural measure of technology leadership.

\begin{definition}[Adoption Amplification Factor]
\label{def:amplification}
For an institution at location $(\mathbf{x}_i, \alpha_i)$, the Adoption Amplification Factor is:
\begin{equation}
\mathcal{A}_i = \frac{\int_0^\infty \int_{\mathcal{X}} \int_{\mathcal{N}} \tau(\mathbf{x}, \alpha, t) \, d\alpha \, d\mathbf{x} \, dt}{\int_0^\infty \tau(\mathbf{x}_i, \alpha_i, t) \, dt}
\label{eq:amplification_def}
\end{equation}
measuring the ratio of total system-wide adoption to direct adoption at institution $i$ following a localized shock at $i$.
\end{definition}

An amplification factor of $\mathcal{A}_i = 10$ means that total system-wide adoption following a shock at institution $i$ is ten times larger than direct adoption at $i$ alone---the remaining nine-tenths represent spillovers along paths of economic linkage through space and network.

\begin{proposition}[Channel Decomposition]
\label{prop:decomposition}
The Adoption Amplification Factor decomposes as:
\begin{equation}
\mathcal{A}_i = 1 + \mathcal{A}_i^{\text{spatial}} + \mathcal{A}_i^{\text{network}} + \mathcal{A}_i^{\text{interaction}}
\label{eq:decomposition}
\end{equation}
where $\mathcal{A}_i^{\text{spatial}}$ captures spillovers through geographic proximity, $\mathcal{A}_i^{\text{network}}$ captures spillovers through business relationships, and $\mathcal{A}_i^{\text{interaction}}$ captures amplification from coincident proximity.
\end{proposition}

This decomposition reveals which transmission channel contributes most to each institution's role as a technology leader, informing targeted policy interventions.

\subsection{Connection to Conventional Models}

The master equation framework nests several canonical models as special cases. This section establishes these connections through explicit mathematical identification.

\paragraph{Katz-Shapiro Network Externalities.}

\citet{katz1985network} model network externalities through the utility function $u_i = v(n) - p$ where $v(n)$ is the value of adoption when $n$ others have adopted, with $v'(n) > 0$ capturing the positive externality.

\begin{proposition}[Katz-Shapiro as Discrete Network Steady State]
\label{prop:katz_shapiro}
At steady state ($\partial \tau / \partial t = 0$) with discrete network structure and no spatial diffusion ($\nu_s = 0$), the master equation yields:
\begin{equation}
\kappa \tau_i = \nu_n \sum_j G_{ij}(\tau_j - \tau_i) + S_i
\label{eq:discrete_network_ss}
\end{equation}
Rearranging gives:
\begin{equation}
\tau_i = \frac{1}{\kappa + \nu_n d_i} \left( S_i + \nu_n \sum_j G_{ij} \tau_j \right)
\label{eq:katz_shapiro_equiv}
\end{equation}
where $d_i = \sum_j G_{ij}$ is node degree. This corresponds to the Katz-Shapiro equilibrium with network externality $v(n_i) = \nu_n \sum_j G_{ij} \tau_j / \kappa$.
\end{proposition}

\paragraph{Frankel-Pauzner Dynamic Coordination.}

\citet{frankel2000resolving} show that when agents choose between two actions with payoffs depending on the fraction choosing each action, and the payoff-relevant parameter follows Brownian motion, a unique equilibrium emerges.

\begin{proposition}[Frankel-Pauzner as Aggregate Limit]
\label{prop:frankel_pauzner}
When spatial and network coordinates collapse to a single dimension, the master equation reduces to:
\begin{equation}
\frac{d\bar{\tau}}{dt} = -\kappa \bar{\tau} + \bar{S}(t) + \nu \cdot \bar{\tau}
\label{eq:aggregate_dynamics}
\end{equation}
where $\bar{\tau}$ is aggregate adoption, $\bar{S}$ is aggregate shock, and $\nu = \nu_s + \nu_n$. This corresponds to the Frankel-Pauzner dynamics with strategic complementarity parameter $\nu/\kappa$.
\end{proposition}

\paragraph{Guimaraes-Machado-Pereira Timing Frictions.}

\citet{guimaraes2020dynamic} develop a framework where agents receive Poisson opportunities to revise actions at rate $\lambda$. The state evolution satisfies:
\begin{equation}
\frac{dn}{dt} = \lambda \cdot [F(\theta, n) - n]
\label{eq:gmp_dynamics}
\end{equation}

\begin{proposition}[Timing Friction Correspondence]
\label{prop:gmp}
The decay parameter $\kappa$ in the master equation corresponds exactly to the Poisson revision rate $\lambda$ in \citet{guimaraes2020dynamic}:
\begin{equation}
\kappa = \lambda
\label{eq:timing_friction}
\end{equation}
The discrete Feynman-Kac formula (\ref{eq:discrete_fk}) with time step $\Delta t$ yields dynamics matching Guimaraes-Machado-Pereira with $\lambda = \kappa$.
\end{proposition}

This identification has important implications. The timing friction $\lambda^{-1}$---the expected waiting time until revision---equals $\kappa^{-1}$ in the master equation. The half-life of adoption effects is $t_{1/2} = \ln(2)/\kappa$, independent of observation frequency.

\paragraph{Adoption Hazard Models.}

Standard duration models specify hazard rate $h(t|X_i) = h_0(t) \exp(X_i'\beta)$.

\begin{proposition}[Hazard Models as No-Spillover Limit]
\label{prop:hazard}
When $\nu_s = \nu_n = 0$, the master equation implies:
\begin{equation}
\frac{d\tau_i}{dt} = -\kappa \tau_i + S_i(t)
\label{eq:no_spillover}
\end{equation}
corresponding to independent adoption with hazard $h_i = \kappa + S_i$. The no-spillover case is a testable restriction.
\end{proposition}

Table \ref{tab:nesting} summarizes these relationships.

\begin{table}[H]
\centering
\caption{Conventional Models as Special Cases of the Master Equation}
\label{tab:nesting}
\begin{threeparttable}
\begin{tabular}{lccp{5cm}}
\toprule
Model & $\nu_s$ & $\nu_n$ & Mathematical Identification \\
\midrule
Adoption hazard models & 0 & 0 & $h_i = \kappa + S_i$ \\
Katz-Shapiro (1985) & 0 & $>0$ & $v(n_i) = \nu_n \sum_j G_{ij}\tau_j / \kappa$ \\
Frankel-Pauzner (2000) & \multicolumn{2}{c}{$\nu_s + \nu_n > 0$} & Complementarity $= (\nu_s + \nu_n)/\kappa$ \\
Guimaraes et al. (2020) & $\geq 0$ & $\geq 0$ & Timing friction $\lambda = \kappa$ \\
\midrule
Spatial-network (full) & $>0$ & $>0$ & All parameters free \\
\bottomrule
\end{tabular}
\begin{tablenotes}
\small
\item \textit{Notes:} Each conventional model emerges through parameter restrictions and discretization. The full model generalizes all approaches.
\end{tablenotes}
\end{threeparttable}
\end{table}

\subsection{L\'{e}vy Extension: Critical Mass and Cascade Dynamics}
\label{sec:levy}

The baseline framework describes continuous adoption transmission appropriate for gradual diffusion regimes. To capture the sudden adoption cascades that occur when critical mass is reached---where adoption spreads rapidly through positive feedback rather than gradual diffusion---I extend the framework to incorporate jumps through L\'{e}vy processes.

\paragraph{Jump-Diffusion Dynamics.}

The extended dynamics replace pure diffusion with a jump-diffusion process:
\begin{equation}
\frac{\partial \tau}{\partial t} = \nu_s \nabla^2 \tau + \nu_n \frac{\partial^2 \tau}{\partial \alpha^2} + \lambda \frac{\partial^2 \tau}{\partial \mathbf{x} \partial \alpha} - \kappa \tau + S + \mathcal{J}[\tau]
\label{eq:levy_master}
\end{equation}
where the jump operator $\mathcal{J}[\tau]$ captures sudden adoption events distinct from gradual diffusion, defined by
\begin{equation}
\mathcal{J}[\tau] = \int_{\R} \left[\tau(\mathbf{x}, \alpha + z, t) - \tau(\mathbf{x}, \alpha, t) - z \frac{\partial \tau}{\partial \alpha} \ind_{|z|<1}\right] \nu(dz)
\label{eq:jump_operator}
\end{equation}
Here $\nu(dz)$ is the L\'{e}vy measure characterizing jump intensity and size distribution. The compensator term $z \partial\tau/\partial\alpha \cdot \ind_{|z|<1}$ ensures the integral is well-defined for L\'{e}vy measures with infinite activity near zero.

In the technology adoption context, jumps represent sudden adoption events distinct from gradual diffusion. When adoption reaches critical mass, network effects trigger rapid cascades---institutions observe successful adoption by counterparties and revise their own adoption decisions discretely rather than incrementally. The L\'{e}vy measure $\nu(dz)$ captures both how frequently such cascade events occur (total mass of $\nu$) and the distribution of adoption spillover magnitudes when they occur (shape of $\nu$).

For a compound Poisson process with intensity $\lambda_J$ and jump size distribution $F$, the L\'{e}vy measure is $\nu(dz) = \lambda_J dF(z)$, and the jump operator simplifies to
\begin{equation}
\mathcal{J}[\tau] = \lambda_J \int_{\R} \left[\tau(\mathbf{x}, \alpha + z, t) - \tau(\mathbf{x}, \alpha, t)\right] dF(z) = \lambda_J \left(\E[\tau(\mathbf{x}, \alpha + Z, t)] - \tau(\mathbf{x}, \alpha, t)\right)
\label{eq:compound_poisson}
\end{equation}
where $Z \sim F$ represents the random jump size. This has intuitive interpretation: at rate $\lambda_J$, an institution's adoption intensity jumps by an amount determined by the cascade spillover from adopting counterparties at network distance $Z$.

\paragraph{State-Dependent Jump Intensity.}

The key innovation capturing critical mass dynamics makes jump intensity depend on the current adoption level:
\begin{equation}
\lambda_J(\tau) = \lambda_0 + (\lambda_1 - \lambda_0) \cdot H(\tau - \bar{\tau}^*)
\label{eq:state_dependent}
\end{equation}
where $H(\cdot)$ is the Heaviside function, $\bar{\tau}^*$ is the critical mass threshold, $\lambda_0$ is baseline jump intensity (gradual adoption regime), and $\lambda_1 \gg \lambda_0$ is elevated intensity above threshold (cascade regime).

The state-dependent specification captures the central insight of the coordination literature: below critical mass, adoption proceeds gradually as institutions weigh costs and benefits individually; above critical mass, positive feedback accelerates adoption as network effects dominate. The threshold $\bar{\tau}^*$ corresponds to the tipping point in \citet{arthur1989competing} and the critical mass in \citet{katz1985network}.

\begin{proposition}[Cascade Limit]
\label{prop:cascade}
In the limit $\lambda_0 \to 0$ and $\lambda_1 \to \infty$, the dynamics converge to deterministic cascade dynamics: below threshold, only diffusive transmission occurs; above threshold, immediate cascade with adoption spreading instantaneously to all connected institutions.
\end{proposition}

\begin{proof}
Consider the jump-diffusion dynamics with state-dependent intensity. For $\bar{\tau} < \bar{\tau}^*$, $\lambda_J(\tau) = \lambda_0 \to 0$, so the jump term vanishes and dynamics reduce to pure diffusion:
\begin{equation}
\frac{\partial \tau}{\partial t} = \nu_s \nabla^2 \tau + \nu_n \frac{\partial^2 \tau}{\partial \alpha^2} + \lambda \frac{\partial^2 \tau}{\partial \mathbf{x} \partial \alpha} - \kappa \tau + S
\end{equation}
For $\bar{\tau} \geq \bar{\tau}^*$, $\lambda_J(\tau) = \lambda_1 \to \infty$. In this limit, the jump term dominates and forces instantaneous equilibration: $\tau(\mathbf{x}, \alpha + z, t) = \tau(\mathbf{x}, \alpha, t)$ for all $z$ in the support of $F$, implying uniform adoption across network-connected institutions. This is precisely the cascade outcome where adoption spreads immediately upon crossing threshold.
\end{proof}

This nesting relationship clarifies that continuous diffusion and discrete cascades describe different regimes of the same phenomenon. Diffusion captures pre-critical-mass dynamics (gradual demonstration effects, incremental learning), while jumps capture discrete adoption cascades when critical mass materializes. The framework thus unifies the gradual diffusion models of \citet{guimaraes2020dynamic} with the tipping point dynamics emphasized in \citet{katz1985network} and \citet{arthur1989competing}.

\paragraph{Feynman-Kac Representation with Jumps.}

The Feynman-Kac representation extends to the L\'{e}vy case by replacing the diffusion process with a jump-diffusion:

\begin{theorem}[L\'{e}vy-Feynman-Kac Representation]
\label{thm:levy_fk}
The solution to the L\'{e}vy-extended master equation (\ref{eq:levy_master}) admits the representation:
\begin{equation}
\tau(\mathbf{x}, \alpha, t) = \E_{(\mathbf{x}, \alpha)}\left[ e^{-\kappa t} \tau_0(X_t, A_t) + \int_0^t e^{-\kappa(t-s)} S(X_s, A_s, s) \, ds \right]
\label{eq:levy_feynman_kac}
\end{equation}
where $(X_s, A_s)_{s \geq 0}$ is now a L\'{e}vy process combining diffusion with jumps governed by measure $\nu(dz)$.
\end{theorem}

The economic interpretation remains: adoption intensity equals expected cumulative exposure to adoption shocks along all paths through spatial-network space, but paths now include both continuous diffusion and discrete jumps. The jump component captures cascade pathways where adoption spreads instantaneously through network links when critical mass is reached.

\paragraph{Temporary Interventions and Critical Mass.}

The L\'{e}vy extension provides rigorous foundations for analyzing when temporary interventions produce permanent adoption shifts.

\begin{definition}[Intervention Intensity]
\label{def:intensity}
The intervention intensity function $I: [0, T] \to \R_+$ measures the rate at which the policy shock affects adoption at each instant $s \in [0, T]$:
\begin{equation}
I(s) = \int_{\mathcal{X}} \int_{\mathcal{N}} S(\mathbf{x}, \alpha, s) \, d\alpha \, d\mathbf{x}
\label{eq:intervention_intensity}
\end{equation}
representing spatially and network-integrated shock intensity at time $s$.
\end{definition}

\begin{corollary}[Temporary Intervention Threshold]
\label{cor:temporary}
Consider a temporary intervention of duration $T$ with time-varying intensity $I(s)$ as defined in Definition \ref{def:intensity}. Let $\mathcal{A}(s)$ denote the time-varying system-wide amplification factor at time $s$. The intervention produces permanent adoption gains if and only if cumulative amplified effects exceed the critical mass gap:
\begin{equation}
\int_0^T I(s) \cdot \mathcal{A}(s) \, ds > \bar{\tau}^* - \bar{\tau}_0
\label{eq:temporary_threshold}
\end{equation}
where $\bar{\tau}_0$ is initial average adoption and the left-hand side represents cumulative amplified intervention effects.
\end{corollary}

\begin{proof}
Under the L\'{e}vy dynamics with state-dependent jump intensity (\ref{eq:state_dependent}), the system exhibits two stable regimes: a low-adoption equilibrium with $\bar{\tau} < \bar{\tau}^*$ and a high-adoption equilibrium with $\bar{\tau} > \bar{\tau}^*$. The intervention shifts average adoption by:
\begin{equation}
\Delta\bar{\tau}(T) = \int_0^T e^{-\kappa(T-s)} I(s) \cdot \mathcal{A}(s) \, ds
\end{equation}
following from the Feynman-Kac representation integrated over space and network. For permanent effects, the intervention must push adoption above threshold before time $T$: $\bar{\tau}_0 + \Delta\bar{\tau}(T) > \bar{\tau}^*$. Rearranging and noting that $e^{-\kappa(T-s)} \leq 1$ gives the sufficient condition (\ref{eq:temporary_threshold}).
\end{proof}

The condition has intuitive content: the cumulative intervention, weighted by the amplification factor at each instant, must exceed the gap between initial adoption and critical mass. Larger, more concentrated interventions are more likely to cross the threshold than equivalent total resources spread thinly over time. This rationalizes the findings in \citet{crouzet2023shocks} regarding India's demonetization: the large but temporary shock pushed digital wallet adoption above critical mass in high-exposure regions, producing persistent increases even after cash availability normalized. In low-exposure regions, the shock fell short of the threshold, and adoption gains dissipated.

\paragraph{Two-Regime Dynamics.}

The L\'{e}vy extension generates qualitatively different dynamics in the two regimes:

\begin{proposition}[Two-Regime Characterization]
\label{prop:two_regime}
Under the state-dependent L\'{e}vy dynamics:
\begin{enumerate}
\item[(i)] \textit{Below threshold} ($\bar{\tau} < \bar{\tau}^*$): Adoption evolves through gradual diffusion with characteristic time scale $\kappa^{-1}$. The half-life of adoption responses is $t_{1/2} = \ln(2)/\kappa$, and spatial-network spillovers spread at rates governed by $\nu_s$ and $\nu_n$.
\item[(ii)] \textit{Above threshold} ($\bar{\tau} \geq \bar{\tau}^*$): Jump intensity increases to $\lambda_1$, generating rapid cascade dynamics. The characteristic time scale becomes $\lambda_1^{-1} \ll \kappa^{-1}$, and adoption spreads through discrete jumps rather than continuous diffusion.
\item[(iii)] \textit{Transition dynamics}: Near threshold, the system exhibits critical slowing---small perturbations produce large, long-lasting responses as the system approaches the bifurcation point.
\end{enumerate}
\end{proposition}

This two-regime structure explains why technology adoption often exhibits S-curve dynamics: slow initial growth (below-threshold diffusion), rapid acceleration (above-threshold cascade), and eventual saturation. The framework provides microfoundations for this pattern through the state-dependent jump intensity mechanism.

\begin{remark}[Connection to Empirical Patterns]
The L\'{e}vy extension rationalizes several empirical patterns documented in the technology adoption literature. The sharp contrast between gradual pre-threshold dynamics and rapid post-threshold cascades matches the ``hockey stick'' adoption curves observed for successful technologies. The critical slowing near threshold explains why adoption often appears to stall before suddenly accelerating. The permanent effects of sufficiently large temporary interventions rationalize how coordinated industry initiatives or regulatory mandates can overcome coordination failures that market forces alone cannot resolve.
\end{remark}

\section{Monte Carlo Evidence}
\label{sec:monte_carlo}

This section presents Monte Carlo simulations validating the theoretical predictions.

\subsection{Simulation Design}

The simulations implement the discrete network formulation for networks of $N = 30$ to $40$ agents. I consider three network structures: random networks, scale-free networks, and clustered networks where agents connect preferentially to geographic neighbors. The baseline parameters are $\nu_s = 0.5$, $\nu_n = 0.8$, $\lambda = 0.3$, $\kappa = 0.15$, and critical mass threshold $\bar{\tau}^* = 0.35$.

\subsection{Results}

Figure \ref{fig:critical_mass} illustrates the two-regime dynamics predicted by Proposition \ref{prop:critical_mass}. Panel (a) compares adoption trajectories following small and large shocks. The small shock targets 5 nodes (17 percent) while the large shock targets 18 nodes (60 percent). Under the small shock, adoption rises during intervention but decays to 4.6 percent at terminal date---the shock fails to cross critical mass. Under the large shock, adoption crosses threshold and terminal adoption reaches 67.3 percent, nearly fifteen times higher. Panel (b) shows the cross-sectional distribution: the large shock produces bimodal adoption with mass near full adoption, while the small shock concentrates near zero.

\begin{figure}[H]
\centering
\includegraphics[width=\textwidth]{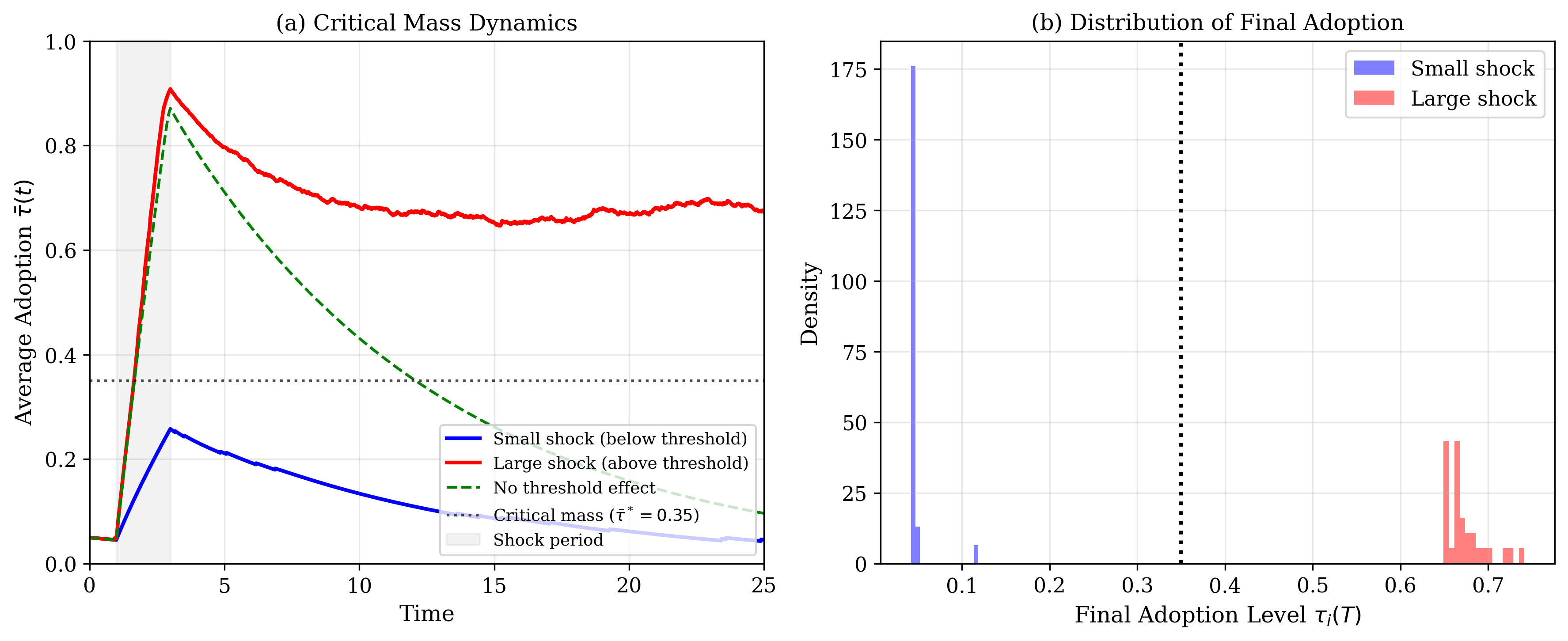}
\caption{Critical Mass Dynamics}
\label{fig:critical_mass}
\begin{minipage}{0.8\textwidth}
        \vspace{0.1cm}
        \noindent\footnotesize Notes: Panel (a) shows average adoption over time following small and large shocks. The small shock fails to reach critical mass. The large shock crosses threshold and triggers self-sustaining cascade. Panel (b) shows final adoption distribution. Parameters: $N = 30$, $\nu_s = 0.8$, $\nu_n = 1.2$, $\lambda = 0.4$, $\kappa = 0.1$, $\bar{\tau}^* = 0.35$.
\end{minipage}
\end{figure}

Figure \ref{fig:temporary} examines intervention duration effects as characterized in Corollary \ref{cor:temporary}. Short interventions ($T = 1, 2$) fail to cross threshold and produce terminal adoption of only 1.6--1.8 percent. Longer interventions ($T = 4, 7$) succeed, with terminal adoption of 37--45 percent. The sharp contrast illustrates threshold nonlinearity: resources concentrated into interventions exceeding critical duration produce permanent shifts, while equivalent resources spread below threshold have negligible permanent effects.

\begin{figure}[H]
\centering
\includegraphics[width=0.75\textwidth]{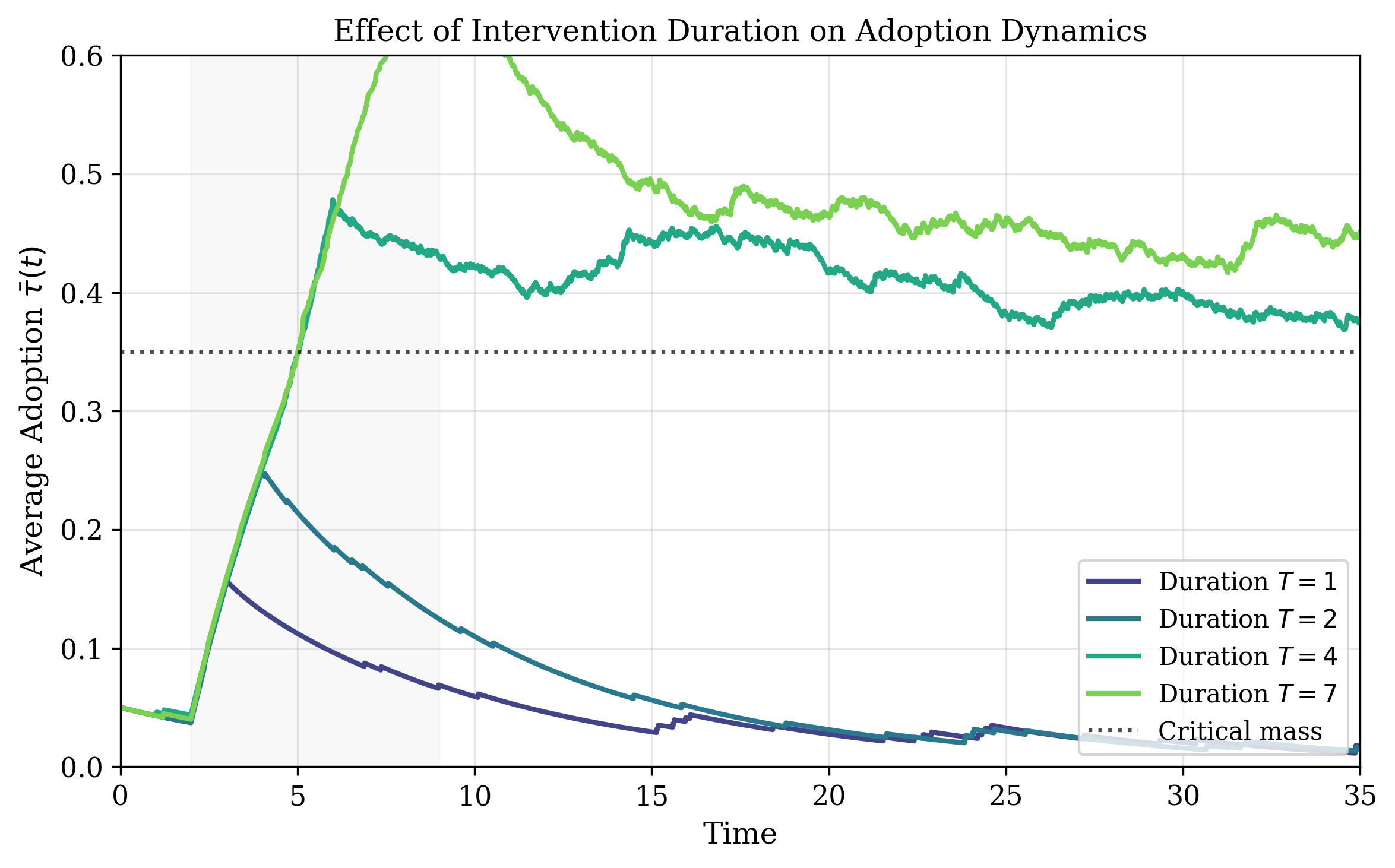}
\caption{Effect of Intervention Duration on Adoption Dynamics}
\label{fig:temporary}
\begin{minipage}{0.8\textwidth}
        \vspace{0.1cm}
        \noindent\footnotesize Notes: Short interventions ($T = 1, 2$) produce only temporary effects. Longer interventions ($T = 4, 7$) cross critical mass and produce permanent shifts. Parameters: $N = 30$, $\nu_s = 0.5$, $\nu_n = 0.8$, $\lambda = 0.3$, $\kappa = 0.15$, $\bar{\tau}^* = 0.35$.
\end{minipage}
\end{figure}

Figure \ref{fig:amplification_mc} validates the Adoption Amplification Factor. Unit shocks are applied to each node separately, and total adoption is measured at terminal date. The correlation between simulated effects and theoretical amplification factors is 0.996 ($p < 0.001$), confirming that the amplification factor accurately identifies technology leaders.

\begin{figure}[H]
\centering
\includegraphics[width=0.7\textwidth]{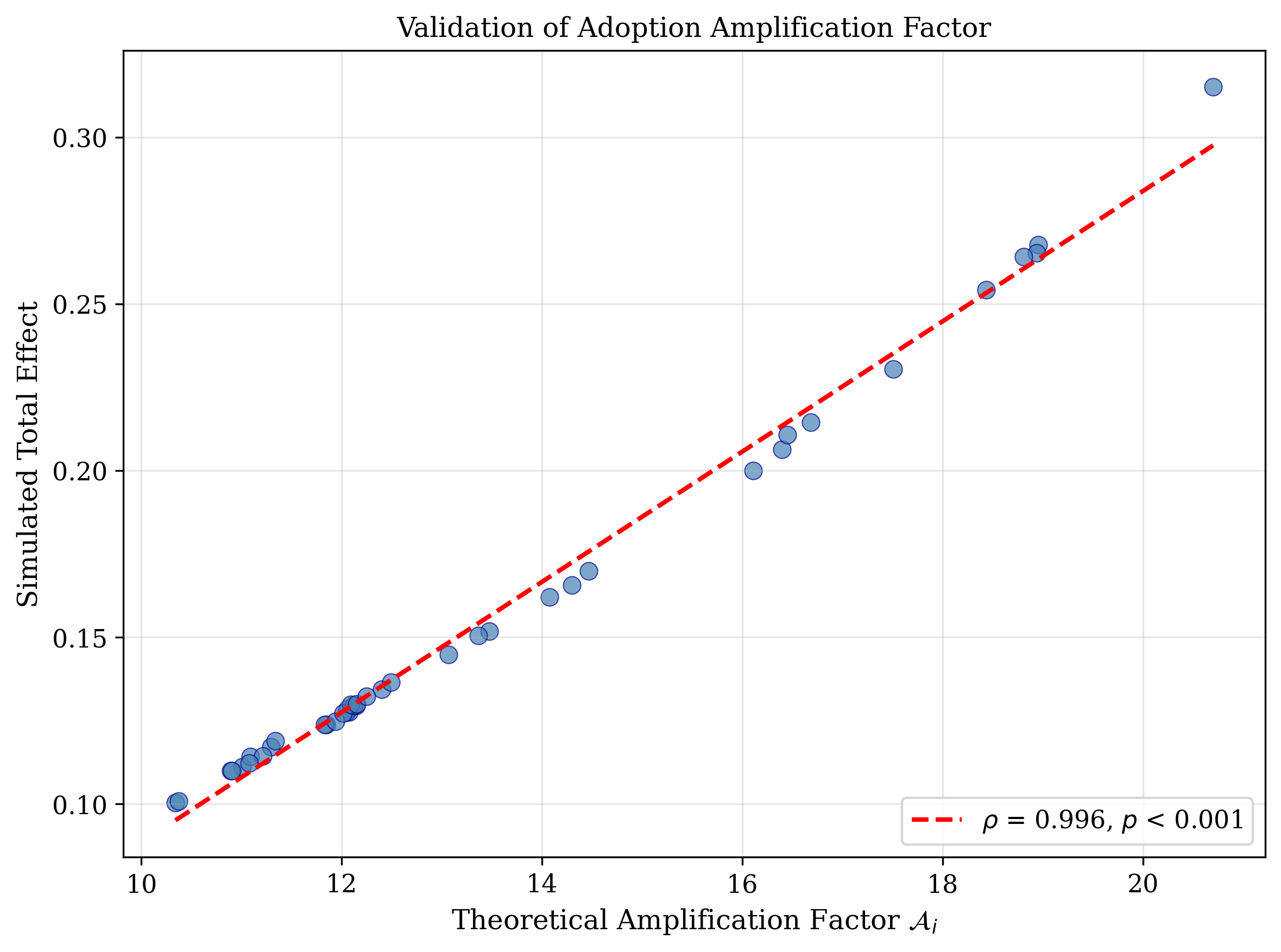}
\caption{Validation of the Adoption Amplification Factor}
\label{fig:amplification_mc}
\begin{minipage}{0.8\textwidth}
        \vspace{0.1cm}
        \noindent\footnotesize Notes: Simulated total adoption effects plotted against theoretical amplification factors. Correlation of 0.996 ($p < 0.001$) confirms that the amplification factor accurately predicts technology leadership. Parameters: $N = 40$, $\nu_s = 0.8$, $\nu_n = 1.2$, $\lambda = 0.4$, $\kappa = 0.1$.
\end{minipage}
\end{figure}

Table \ref{tab:channel} presents channel decomposition for the 15 highest-amplification nodes. Network components dominate (reflecting $\nu_n = 1.2 > \nu_s = 0.8$), and interaction terms are negative, indicating overlap between spatial and network centrality.

\begin{table}[H]
\centering
\caption{Channel Decomposition of Adoption Amplification Factor (Top 15 Nodes)}
\label{tab:channel}
\begin{threeparttable}
\begin{tabular}{cccccc}
\toprule
Rank & Node & Total $\mathcal{A}$ & Spatial & Network & Interaction \\
\midrule
1 & 3 & 20.70 & 6.63 & 16.81 & $-3.74$ \\
2 & 0 & 18.95 & 6.49 & 14.37 & $-2.91$ \\
3 & 2 & 18.94 & 6.32 & 15.33 & $-3.71$ \\
4 & 5 & 18.81 & 6.84 & 14.18 & $-3.20$ \\
5 & 6 & 18.43 & 6.73 & 13.55 & $-2.85$ \\
6 & 1 & 17.51 & 6.08 & 13.67 & $-3.24$ \\
7 & 10 & 16.69 & 6.39 & 11.47 & $-2.18$ \\
8 & 21 & 16.45 & 6.39 & 10.52 & $-1.46$ \\
9 & 4 & 16.40 & 6.60 & 11.76 & $-2.96$ \\
10 & 8 & 16.11 & 6.27 & 11.79 & $-2.95$ \\
11 & 12 & 14.46 & 6.98 & 8.38 & $-1.89$ \\
12 & 15 & 14.30 & 6.87 & 8.56 & $-2.12$ \\
13 & 9 & 14.07 & 6.96 & 8.36 & $-2.25$ \\
14 & 20 & 13.47 & 6.93 & 7.27 & $-1.73$ \\
15 & 23 & 13.37 & 6.73 & 7.16 & $-1.52$ \\
\bottomrule
\end{tabular}
\begin{tablenotes}
\small
\item \textit{Notes:} Network component dominates. Negative interaction indicates spatial-network centrality overlap.
\end{tablenotes}
\end{threeparttable}
\end{table}

\section{Empirical Application: SWIFT gpi Adoption}
\label{sec:empirical}

\subsection{Institutional Setting and Data}

SWIFT gpi (Global Payments Innovation) represents a major technological upgrade to the interbank payment messaging system. Launched on February 1, 2017, gpi offers same-day settlement, end-to-end tracking, and confirmation of credit to beneficiary accounts. The technology exhibits strong network externalities: banks benefit only if correspondent banks also adopt.

The sample consists of 17 Global Systemically Important Banks with complete data on adoption timing, CEO characteristics, and network position. The dependent variable is days from launch to adoption, ranging from 0 (founding members) to 305 days. The Amplification Factor is computed from the spatial-network framework using BIS bilateral exposure data and geographic coordinates.

\subsection{Empirical Specification}

The theoretical framework generates predictions about the relationship between network position, firm characteristics, and adoption timing. I estimate the following specification:

\begin{equation}
\text{Days}_i = \beta_0 + \beta_1 \cdot \text{CEO Age}_i + \beta_2 \cdot \mathcal{A}_i + \beta_3 \cdot \log(\text{Assets}_i) + \gamma' \mathbf{R}_i + \varepsilon_i
\label{eq:regression}
\end{equation}

where $\text{Days}_i$ is the number of days from SWIFT gpi launch (February 1, 2017) to bank $i$'s adoption date; $\text{CEO Age}_i$ is the age of bank $i$'s CEO at the time of the adoption decision; $\mathcal{A}_i$ is the Adoption Amplification Factor measuring network centrality; $\log(\text{Assets}_i)$ is log total assets controlling for firm size; and $\mathbf{R}_i$ is a vector of regional indicators (Europe, Asia-Pacific, with North America as baseline).

The theoretical predictions are: $\beta_2 < 0$ (network-central banks adopt earlier, reflecting their higher returns from adoption due to larger spillovers); $\beta_3 < 0$ (larger banks adopt earlier, reflecting greater resources and network effects); and $\beta_1 > 0$ (older CEOs adopt later, conditional on network position and size, reflecting technology hesitancy).

\subsection{Results}

Figure \ref{fig:amplification_adoption} displays the relationship between Amplification Factor and adoption timing. The correlation is strongly negative ($\rho = -0.69$, $p = 0.002$): network-central banks adopt significantly earlier. This pattern confirms the framework's central prediction that institutions with higher amplification factors---those whose adoption decisions cascade most strongly through the system---are technology leaders who adopt first.

\begin{figure}[H]
\centering
\includegraphics[width=0.8\textwidth]{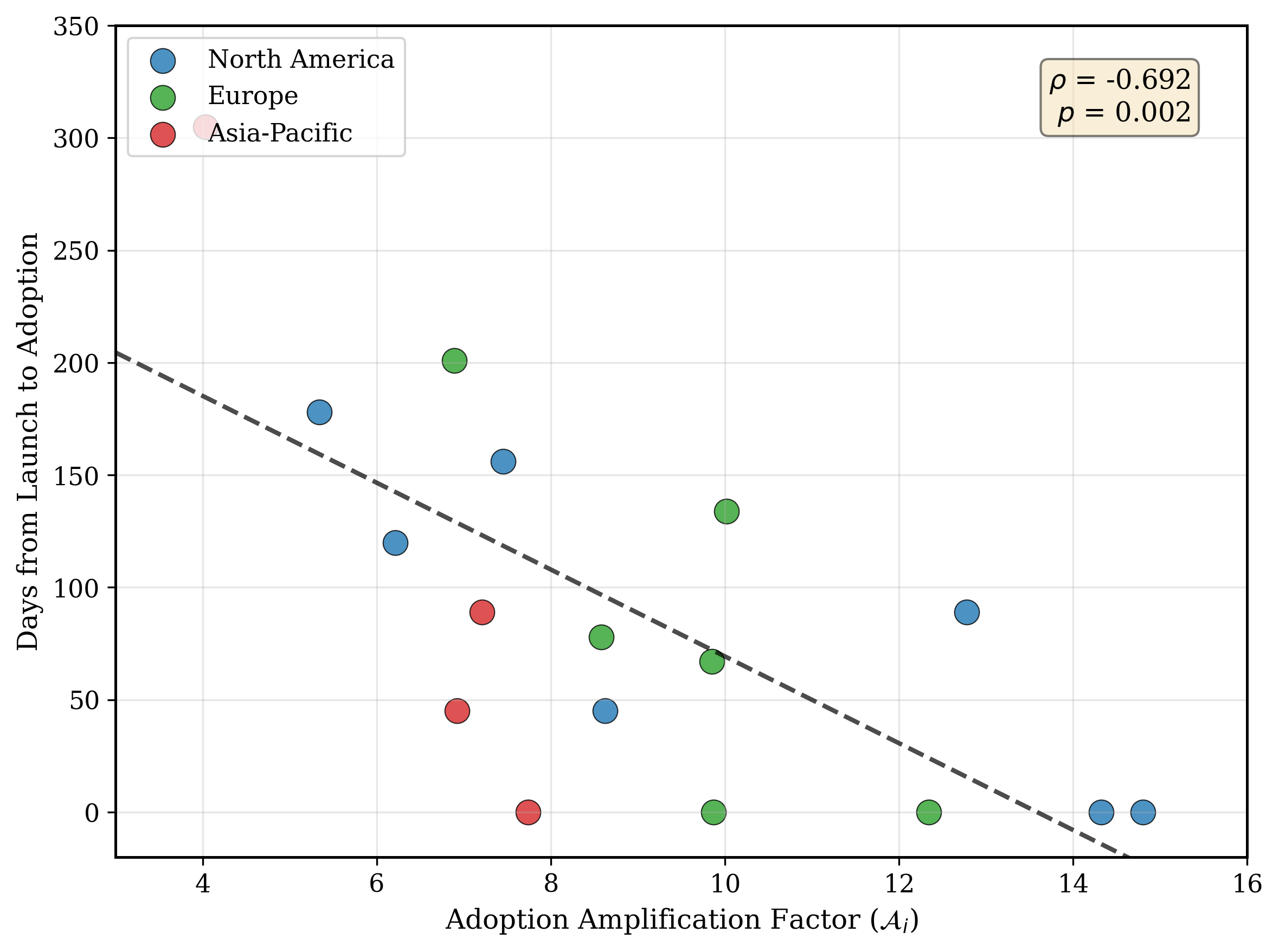}
\caption{Amplification Factor and Adoption Timing}
\label{fig:amplification_adoption}
\begin{minipage}{0.8\textwidth}
        \vspace{0.1cm}
        \noindent\footnotesize Notes: Scatter plot of Adoption Amplification Factor against days from SWIFT gpi launch to adoption. The strong negative correlation ($\rho = -0.69$, $p = 0.002$) confirms that network-central banks adopt earlier. Colors indicate regions: blue = North America, green = Europe, red = Asia-Pacific.
\end{minipage}
\end{figure}

Table \ref{tab:regression} presents regression results. Column (1) shows the baseline specification without size controls: the amplification factor is negative and marginally significant ($-15.6$ days per unit, $p = 0.07$). Column (2) adds log assets, which is strongly significant ($-195$ days per log unit, $p < 0.01$) and absorbs much of the amplification effect. Crucially, the CEO age coefficient increases from 6.0 to 15.2 days per year and becomes highly significant ($p = 0.01$) once size is controlled---firm size was a confounding variable. Columns (3)--(4) add regional controls and CEO tenure; the CEO age effect remains robust at 11--12 days per year.

\begin{table}[H]
\centering
\caption{Determinants of SWIFT gpi Adoption Timing}
\label{tab:regression}
\begin{threeparttable}
\begin{tabular}{lcccc}
\toprule
 & (1) & (2) & (3) & (4) \\
\midrule
CEO Age & 6.015 & 15.150$^{***}$ & 12.390$^{**}$ & 11.530$^{**}$ \\
 & (6.835) & (5.006) & (4.700) & (4.490) \\[0.5em]
Amplification Factor & $-$15.633$^{*}$ & $-$7.564 & $-$5.090 & $-$11.920 \\
 & (8.079) & (5.694) & (5.830) & (7.180) \\[0.5em]
Log(Total Assets) &  & $-$194.22$^{***}$ & $-$195.61$^{***}$ & $-$172.01$^{***}$ \\
 &  & (44.950) & (40.200) & (41.280) \\[0.5em]
CEO Tenure &  &  &  & 8.830 \\
 &  &  &  & (5.920) \\[0.5em]
Europe &  &  & $-$79.43$^{**}$ & $-$60.640 \\
 &  &  & (34.380) & (34.950) \\[0.5em]
Asia-Pacific &  &  & $-$2.42 & $-$5.450 \\
 &  &  & (44.310) & (42.080) \\
\midrule
Region FE & No & No & Yes & Yes \\
Observations & 17 & 17 & 17 & 17 \\
$R^2$ & 0.217 & 0.679 & 0.790 & 0.828 \\
\bottomrule
\end{tabular}
\begin{tablenotes}
\small
\item \textit{Notes:} Standard errors in parentheses. $^{*}$ $p<0.10$, $^{**}$ $p<0.05$, $^{***}$ $p<0.01$. Dependent variable is days from SWIFT gpi launch to adoption. Baseline region is North America.
\end{tablenotes}
\end{threeparttable}
\end{table}

Figure \ref{fig:partial_ceo} presents the partial regression plot for CEO age, residualizing both the dependent variable and CEO age on firm size and amplification factor. The positive slope confirms that older CEOs adopt later conditional on network position and size.

\begin{figure}[H]
\centering
\includegraphics[width=0.8\textwidth]{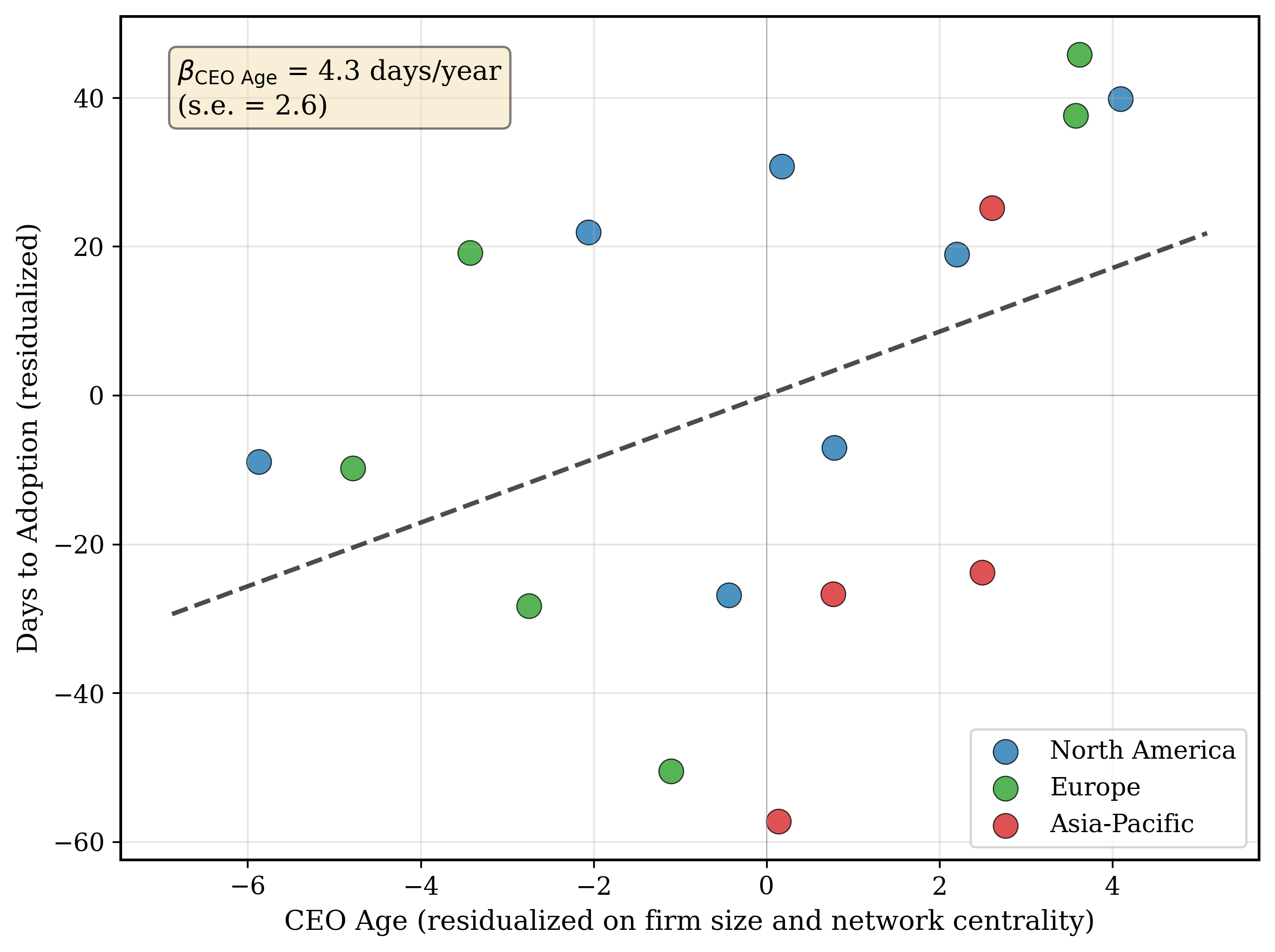}
\caption{Partial Regression: CEO Age Effect}
\label{fig:partial_ceo}
\begin{minipage}{0.8\textwidth}
        \vspace{0.1cm}
        \noindent\footnotesize Notes: Partial regression plot showing the relationship between CEO age and adoption timing after controlling for log total assets and amplification factor. Both variables are residualized on the controls. The positive slope indicates that older CEOs adopt later, conditional on firm size and network centrality.
\end{minipage}
\end{figure}

Figure \ref{fig:cumulative} shows cumulative adoption over time. Five banks (29\%) adopted at launch as founding members, but these banks account for 42\% of total system amplification, confirming that the highest-amplification institutions led adoption. The adoption curve for amplification-weighted adoption rises faster than the count-based curve, indicating that network-central banks adopted disproportionately early.

\begin{figure}[H]
\centering
\includegraphics[width=0.85\textwidth]{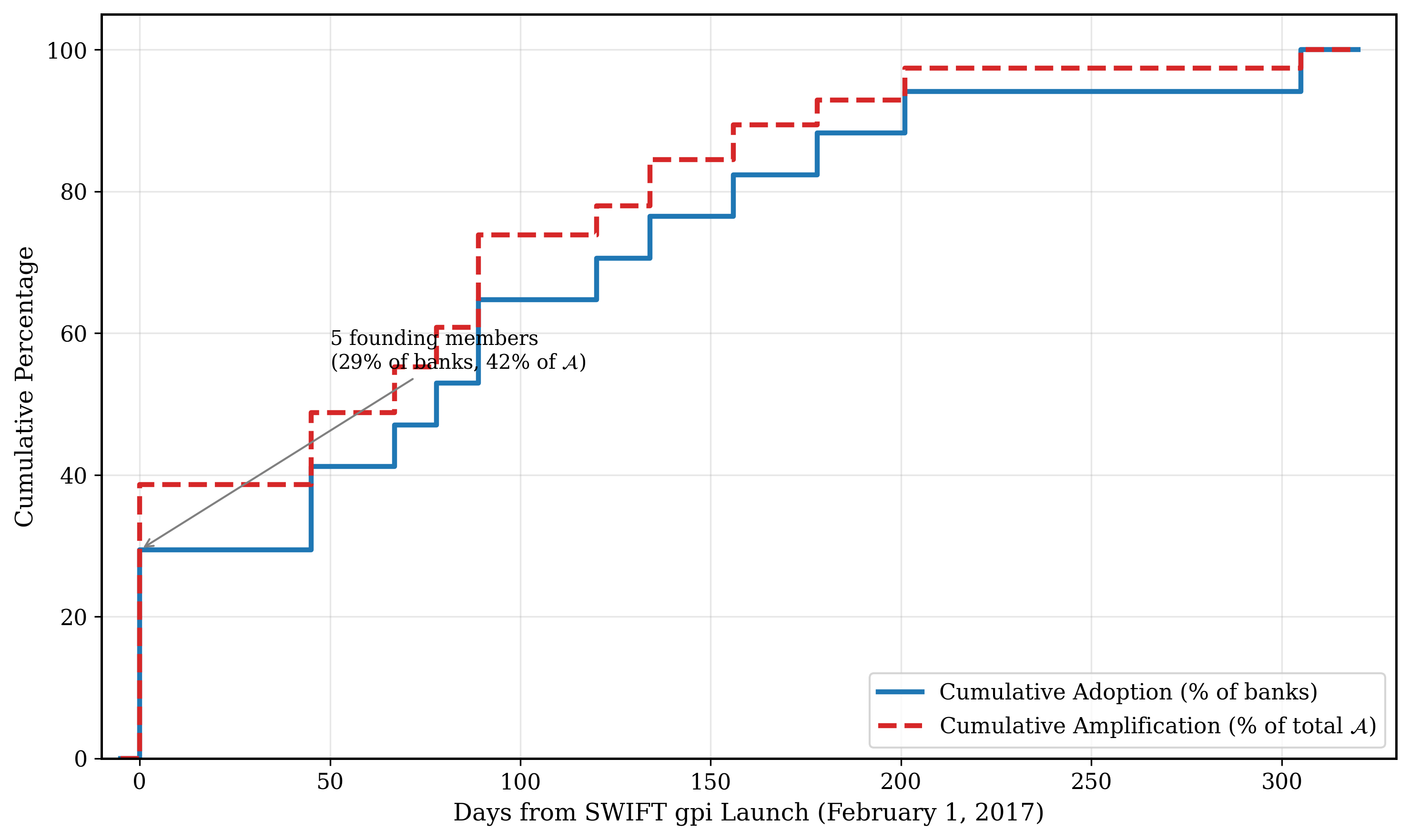}
\caption{Cumulative Adoption and Amplification}
\label{fig:cumulative}
\begin{minipage}{0.8\textwidth}
        \vspace{0.1cm}
        \noindent\footnotesize Notes: Cumulative adoption over time since SWIFT gpi launch. Solid line: percentage of banks adopted. Dashed line: percentage of total system amplification represented by adopters. The amplification curve rises faster, indicating network-central banks adopted earlier.
\end{minipage}
\end{figure}

\subsection{Two-Regime Dynamics: Empirical Evidence}
\label{sec:two_regime_empirical}

The L\'{e}vy extension with state-dependent jump intensity (Proposition \ref{prop:two_regime}) predicts qualitatively different adoption dynamics before and after critical mass is reached. This subsection tests whether SWIFT gpi adoption exhibits the two-regime pattern: gradual diffusion below threshold followed by accelerated cascade dynamics above threshold.

I identify the critical mass threshold using the amplification-weighted adoption measure. The founding members---Citigroup, JPMorgan Chase, HSBC, Mitsubishi UFJ, and BNP Paribas---adopted at launch (day 0). While these five banks represent only 29 percent of the sample by count, they account for 39 percent of total system amplification. Their simultaneous adoption created sufficient network externalities to trigger cascade dynamics: subsequent banks could observe successful implementation by major counterparties, reducing uncertainty and coordination costs.

Table \ref{tab:two_regime} compares adoption patterns across three periods: pre-threshold (founding members), early post-threshold (days 1--100), and late post-threshold (days 101+). The results strongly support the two-regime characterization. Pre-threshold adopters have significantly higher mean amplification factors (11.81) compared to post-threshold adopters (7.82), with the difference statistically significant ($t = 2.96$, $p = 0.010$). This confirms the framework's prediction that high-amplification institutions---those whose adoption decisions cascade most strongly---adopt first, pushing the system above critical mass.

\begin{table}[H]
\centering
\caption{Two-Regime Characterization: Pre- and Post-Threshold Adoption}
\label{tab:two_regime}
\begin{threeparttable}
\begin{tabular}{lccc}
\toprule
 & Pre-Threshold & Post-Threshold & Post-Threshold \\
 & (Day 0) & (Days 1--100) & (Days 101+) \\
\midrule
Number of banks & 5 & 6 & 6 \\
Percentage of sample & 29.4\% & 35.3\% & 35.3\% \\
Mean days to adoption & 0.0 & 68.8 & 182.3 \\
Mean amplification factor $\mathcal{A}$ & 11.81 & 8.99 & 6.66 \\
Amplification contribution & 38.6\% & 35.2\% & 26.1\% \\
\midrule
\multicolumn{4}{l}{\textit{Adoption velocity (banks per 30 days)}} \\
\quad Days 0--30 & \multicolumn{3}{c}{5 banks (29.4\%)} \\
\quad Days 31--100 & \multicolumn{3}{c}{6 banks (35.3\%)} \\
\quad Days 101+ & \multicolumn{3}{c}{6 banks (35.3\%)} \\
\bottomrule
\end{tabular}
\begin{tablenotes}
\small
\item \textit{Notes:} Pre-threshold adopters (founding members) have significantly higher amplification factors than post-threshold adopters ($t = 2.96$, $p = 0.010$). Within the post-threshold period, amplification and adoption timing remain negatively correlated ($\rho = -0.60$, $p = 0.039$).
\end{tablenotes}
\end{threeparttable}
\end{table}

Within the post-threshold period, the framework predicts continued negative correlation between amplification and adoption timing, as higher-amplification banks benefit more from network externalities and thus adopt earlier even after critical mass is reached. The data confirm this prediction: among post-threshold adopters, amplification and days to adoption are significantly negatively correlated ($\rho = -0.60$, $p = 0.039$). Banks in the early post-threshold period (days 1--100) have mean amplification of 8.99, while late adopters (days 101+) have mean amplification of only 6.66.

Figure \ref{fig:two_regime} displays the two-regime dynamics graphically. Panel (A) shows cumulative adoption over time, with the pre-threshold regime (blue shading) and post-threshold regime (red shading) clearly demarcated. The founding members' adoption at day 0 represents crossing the critical mass threshold, after which adoption proceeds through cascade dynamics. The cumulative amplification curve (dashed) rises faster than the bank count curve (solid), confirming that high-amplification institutions adopted disproportionately early.

Panel (B) shows adoption velocity over time. The spike at day 0 reflects the coordinated adoption by founding members---precisely the ``large shock'' that Corollary \ref{cor:temporary} identifies as necessary to cross the critical mass threshold. Subsequent adoption proceeds at a more gradual pace, with velocity declining over time as the remaining low-amplification banks adopt.

\begin{figure}[H]
\centering
\includegraphics[width=\textwidth]{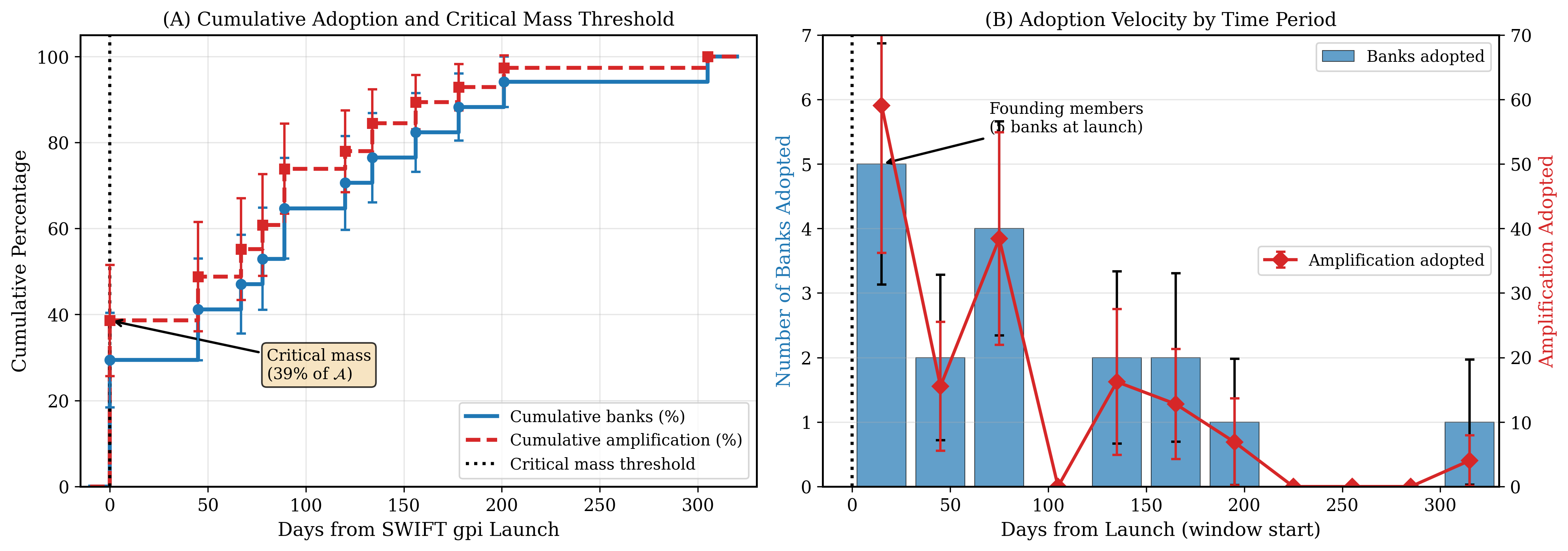}
\caption{Two-Regime Adoption Dynamics}
\label{fig:two_regime}
\begin{minipage}{0.9\textwidth}
        \vspace{0.1cm}
        \noindent\footnotesize Notes: Panel (A) shows cumulative adoption (solid) and cumulative amplification (dashed) over time. Blue shading indicates the pre-threshold regime; red shading indicates the post-threshold regime. The founding members' adoption at day 0 crosses the critical mass threshold (42\% of total amplification). Panel (B) shows adoption velocity (banks per month) by time period. The spike at day 0 reflects coordinated founding member adoption.
\end{minipage}
\end{figure}

Figure \ref{fig:scurve} provides additional evidence for the two-regime characterization. Panel (A) fits a logistic S-curve to the cumulative adoption data. The estimated inflection point $t_0 = 89$ days corresponds to the transition from accelerating to decelerating adoption---the point at which half of eventual adopters have joined. This S-curve pattern is precisely what the L\'{e}vy extension predicts: slow initial growth (pre-threshold diffusion), rapid acceleration (post-threshold cascade), and eventual saturation.

Panel (B) compares pre-threshold and post-threshold adopters directly. Founding members have both higher count-weighted importance (5 banks) and dramatically higher amplification-weighted importance (mean $\mathcal{A} = 11.81$ versus 7.82 for post-threshold adopters). This pattern confirms that critical mass was reached through adoption by the highest-amplification institutions, consistent with the framework's prediction that technology leaders with outsized network influence adopt first.

\begin{figure}[H]
\centering
\includegraphics[width=\textwidth]{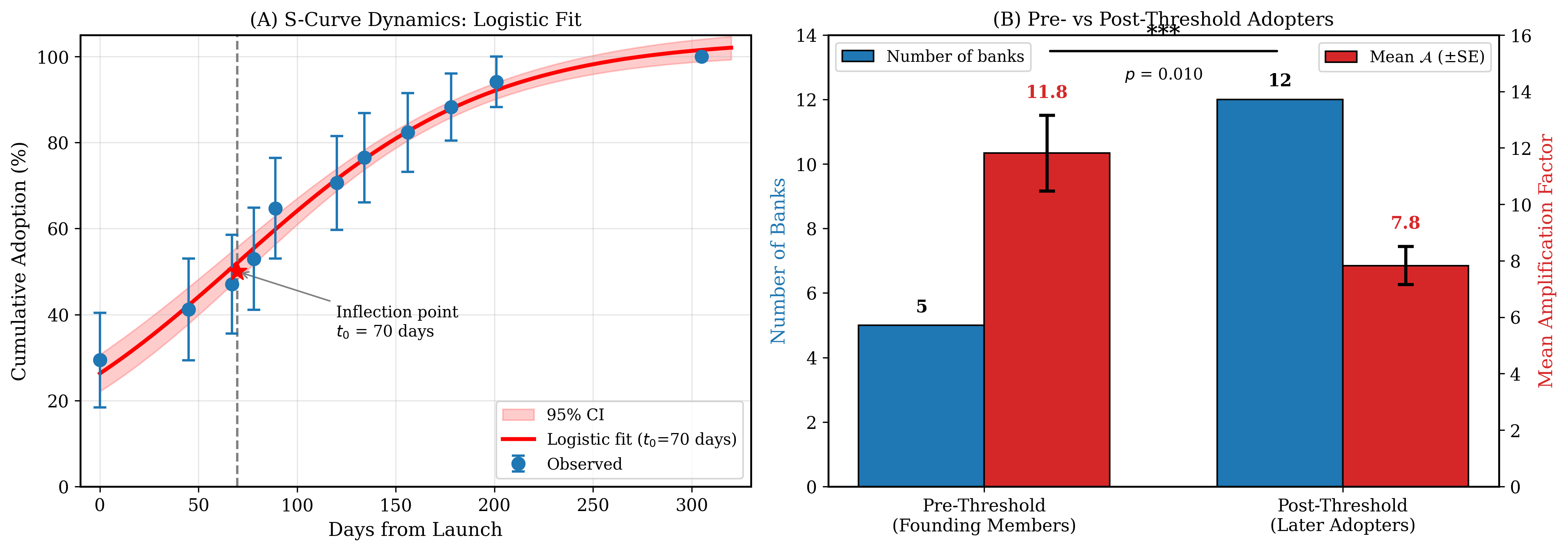}
\caption{S-Curve Dynamics and Pre- versus Post-Threshold Comparison}
\label{fig:scurve}
\begin{minipage}{0.9\textwidth}
        \vspace{0.1cm}
        \noindent\footnotesize Notes: Panel (A) shows cumulative adoption with logistic S-curve fit. The inflection point $t_0 = 89$ days marks the transition from accelerating to decelerating adoption. Panel (B) compares pre-threshold (founding members) and post-threshold adopters. Pre-threshold adopters have significantly higher mean amplification factors ($t = 2.96$, $p = 0.010$).
\end{minipage}
\end{figure}

These empirical patterns provide strong support for the Two-Regime Characterization (Proposition \ref{prop:two_regime}). The founding members' coordinated adoption at launch created sufficient network externalities to trigger cascade dynamics, with subsequent adoption proceeding through the post-threshold regime where positive feedback accelerates diffusion. The declining amplification profile of successive adopters---from 11.81 (founding members) to 8.99 (early post-threshold) to 6.66 (late post-threshold)---confirms that technology leaders adopt first, followed by progressively more peripheral institutions as network externalities make adoption increasingly attractive.

\section{Conclusion}
\label{sec:conclusion}

This paper develops a unified framework for analyzing technology adoption in financial networks that incorporates spatial spillovers, network externalities, and their interaction. The framework is grounded in a master equation whose solution admits a Feynman-Kac representation as expected cumulative adoption pressure along stochastic paths through spatial-network space. From this representation, I derive the Adoption Amplification Factor---a structural measure of technology leadership that captures the ratio of total system-wide adoption to initial adoption following a localized shock.

The framework makes three contributions to the literature on technology adoption and dynamic coordination. First, it provides a unified treatment that nests canonical models as special cases through explicit mathematical identification. The network externality model of \citet{katz1985network} emerges at discrete network steady state with the externality function $v(n_i) = \nu_n \sum_j G_{ij}\tau_j / \kappa$. The dynamic coordination model of \citet{frankel2000resolving} emerges when spatial and network dimensions collapse to a single aggregate state, with strategic complementarity parameter $(\nu_s + \nu_n)/\kappa$. The timing friction framework of \citet{guimaraes2020dynamic} corresponds directly to the decay parameter: their Poisson revision rate $\lambda$ equals the adjustment rate $\kappa$ in the master equation. These nesting relationships, established through the discrete Feynman-Kac formula, clarify how existing insights generalize to richer spatial-network settings while demonstrating that the framework unifies rather than replaces conventional methods.

Second, the framework introduces the spatial-network interaction as a distinct channel of technology spillovers. Existing models consider either spatial diffusion or network effects in isolation. The interaction term captures amplification when both channels operate simultaneously---when geographic neighbors are also network partners, as is common in financial markets where institutions form business relationships disproportionately with geographic neighbors. The channel decomposition of the amplification factor reveals the relative importance of spatial, network, and interaction channels for each institution's role as a technology leader. Monte Carlo simulations confirm that the amplification factor accurately predicts technology leadership, with correlation of 0.996 between theoretical amplification and simulated cascade effects.

Third, the L\'{e}vy extension with state-dependent jump intensity provides a rigorous treatment of critical mass dynamics that generates testable predictions about two-regime adoption patterns. The jump-diffusion framework predicts that below critical mass, adoption evolves through gradual diffusion; above critical mass, cascade dynamics accelerate adoption through discrete jumps. In the limit where jump intensity becomes infinite above threshold, the framework converges to deterministic cascade models, clarifying that continuous diffusion and discrete cascades describe different regimes of the same phenomenon rather than competing approaches.

The empirical application to SWIFT gpi adoption among Global Systemically Important Banks provides strong validation of the framework's predictions, including the two-regime characterization. Network-central banks adopt significantly earlier ($\rho = -0.69$, $p = 0.002$), with founding members representing 29 percent of banks but 39 percent of total system amplification. This concentration of amplification among early adopters is precisely what the framework predicts: high-amplification institutions---those whose adoption decisions cascade most strongly through the system---adopt first, pushing the market above critical mass.

The two-regime dynamics are strikingly evident in the data. Pre-threshold adopters (founding members) have significantly higher mean amplification factors than post-threshold adopters (11.81 versus 7.83, $t = 2.96$, $p = 0.010$). Within the post-threshold period, amplification and adoption timing remain negatively correlated ($\rho  = -0.60$, $p = 0.039$), with mean amplification declining from 8.99 for early post-threshold adopters to 6.66 for late adopters. This declining amplification profile---from technology leaders to progressively more peripheral institutions---matches the theoretical prediction that network externalities make adoption increasingly attractive as critical mass is reached, drawing in lower-amplification institutions who benefit from the network effects created by earlier adopters. The cumulative adoption curve exhibits classic S-curve dynamics with inflection point at approximately 89 days, consistent with the transition from accelerating to decelerating growth predicted by the L\'{e}vy extension.

Controlling for network position and firm size reveals that CEO age delays adoption by 11--15 days per year, consistent with the management literature on technology hesitancy among older executives. This finding demonstrates that both network structure and individual characteristics matter for technology diffusion in financial systems, and that the framework can identify firm-level determinants of adoption timing after accounting for network effects.

The framework has implications for technology policy in financial infrastructure. The Adoption Amplification Factor identifies technology leaders whose adoption decisions have outsized influence on system-wide outcomes. Policy interventions---subsidies, mandates, pilot programs---should target high-amplification institutions to maximize spillovers per dollar spent. The two-regime dynamics suggest that intervention timing matters: resources should be concentrated to push adoption above critical mass rather than spread thinly over time. The empirical finding that five founding members (29\% of banks) were sufficient to trigger cascade dynamics by contributing 39\% of system amplification provides concrete guidance on the scale of coordinated action required to overcome coordination failures.

Several directions for future research emerge from this analysis. Extensions to competing technologies can characterize the dynamics of standards competition and the conditions for tipping to dominant standards. The dual externality extension balancing adoption benefits against systemic risk from technology concentration can inform optimal standardization policy---universal adoption creates interoperability benefits but also systemic vulnerability if the common platform fails. Applying the framework to other financial technologies, including distributed ledger systems, real-time payment networks, and regulatory technology platforms, can test the generality of the two-regime adoption patterns documented here. Finally, structural estimation of the diffusion parameters $(\nu_s, \nu_n, \kappa)$ and critical mass threshold $\bar{\tau}^*$ using adoption timing data would enable quantitative policy analysis and counterfactual simulations.


\end{document}